\documentclass[journal]{IEEEtran}
\usepackage{amsmath}    
\usepackage{amsfonts}
\usepackage{amssymb}
\usepackage{bbm}
\usepackage{here}
\usepackage{mathrsfs}
\usepackage{graphicx}
\usepackage{url}
\usepackage{color}
\usepackage{psfrag}
\usepackage{epsfig}
\usepackage{subfigure}

\newtheorem{theorem}{Theorem}[section]
\newtheorem{lemma}{Lemma}[section]
\newtheorem{corollary}{Corollary}[section]
\newtheorem{proposition}{Proposition}[section]

\newtheorem{remark}{Remark}[section]
\newtheorem{example}{Example}[section]

 {\everymath{\displaystyle\everymath{}}\equation}%
 {\endequation}

\DeclareMathOperator*{\col}{col}

\DeclareMathOperator*{\sgn}{sgn}

\newcommand{\meth}{\hfill $\diamond$}

\renewenvironment{bmatrix}%
{\ensuremath{\left[\begin{array}{ccccccccccccccccccccccccccccccccccccccccccccccccccccccccccccccccccccc}}}{\ensuremath{\end{array}\right]}}

\renewcommand{\paragraph}[1]{\textbf{#1}}

\title{Convergence and Equivalence results for the Jensen's inequality - Application to time-delay and sampled-data systems}

\author{Corentin~Briat\thanks{ACCESS Linnaeus Centre, Division of Optimization and Systems Theory, KTH, SE-100 44, Stockholm, Sweden; corentin@briat.info, cbriat@math.kth.se; http://www.briat.info}\thanks{This work has been supported by the ACCESS and RICSNET projects, KTH, Stockholm, Sweden. http://www.access.kth.se/}}

\begin{document}
\maketitle

\begin{abstract}
The Jensen's inequality plays a crucial role in the analysis of time-delay and sampled-data systems. Its conservatism is studied through the use of the Gr\"{u}ss Inequality. It has been reported in the literature that fragmentation (or partitioning) schemes allow to empirically improve the results. We prove here that the Jensen's gap can be made arbitrarily small provided that the order of uniform fragmentation is chosen sufficiently large. Non-uniform fragmentation schemes are also shown to speed up the convergence in certain cases. Finally, a family of bounds is characterized and a comparison with other bounds of the literature is provided. It is shown that the other bounds are equivalent to Jensen's and that they exhibit interesting well-posedness and linearity properties which can be exploited to obtain better numerical results.
\end{abstract}

\begin{keywords}
 Jensen's Inequality;  Gr\"{u}ss Inequality; Time-Delay Systems; Sampled-data systems; Conservatism; Fragmentation
\end{keywords}

\section{Introduction}

The Jensen's Inequality \cite{Jensen:1905} has had a tremendous impact on many different fields; e.g. convex analysis, probability theory, information theory, statistics, control and systems theory \cite{GuKC:03, Gouaisbaut:06, Briat:08phd}. It concerns the bounding of convex functions of integrals or sums:
\begin{lemma}
  Let $U$ be a given connected and compact set of $\mathbb{R}$, $f$ a function measurable over $U$ and $\phi$ a convex function measurable over $f(U)$. Then the inequality
  \begin{equation}\label{eq:Jensenc}
    \phi\left(\int_Uf(s)d\mu(s)\right)\le\mu(U)\int_U[\phi\circ f](s)d\mu(s)
  \end{equation}
  holds where $\mu$ is a given nonnegative measure, e.g. the Lebesgue measure, and $\mu(U)=\int_Ud\mu(s)<+\infty$ is the measure of the set $U$.
\end{lemma}

The discrete counterpart is given by:
\begin{lemma}
  Let $U$ be a given connected and compact set of $\mathbb{Z}$, $f$ a function measurable over $U$ and $\phi$ a convex function measurable over $f(U)$. Then the inequality
  \begin{equation}\label{eq:Jensend}
    \phi\left(\mu(U)^{-1}\sum_{i\in U}f_i\mu_i\right)\le\mu(U)^{-1}\sum_{i\in U}\phi(f_i)\mu_i
  \end{equation}
  holds where $\{\mu_i\}_{i\in U}$ is a given nonnegative measure, e.g. the counting measure, and $\mu(U)=\sum_{i\in U}\mu_i<+\infty$ is the measure of the set $U$.
\end{lemma}

These inequalities have found applications in systems theory, for instance for the computation of an upper bound on the $\mathcal{L}_2$-gain of integral operators involved in time-delay systems analysis \cite{GuKC:03, Ariba:08, Briat:09f}. Another application in time-delay systems \cite{Han:05, Gouaisbaut:06, Briat:08c, Briat:08phd} concerns the bounding of integral quadratic terms of the form
$-\int_{t-h}^t\dot{x}(s)^TR\dot{x}(s)ds$
arising in some approaches based on Lyapunov-Krasovskii functionals (LKFs). Discrete counterparts, involving sums instead of integrals, also exist; see e.g. \cite{Zhang:06, Kao:08}.

The objective of the paper is twofold. The first goal is to study the conservatism induced by the use of Jensen's inequality. Using the Gr\"{u}ss inequality, bounds on the gap will be obtained in the general case and refined by considering the eventual differentiability of the function $f$. Recent results have reported an empirical conservatism reduction of LKFs when using 'delay-fragmentation' or 'delay-partitioning' procedures \cite{Gouaisbaut:06,Gouaisbaut:06b, Laas:07,Han:08,Briat:08phd}. This, however, remains to be proved theoretically and, as a first step towards this result, we will show that the fragmentation scheme reduces the gap of the Jensen's inequality. It will be also proved that the upper bound on the gap converges sublinearly to 0 as the fragmentation order increases. Finally, non-uniform fragmentation schemes are introduced and their accelerating effect on the convergence is illustrated.

The second objective of the paper is to show the equivalence between several bounds provided in the literature. First, a complete family of bounds is characterized for which the equivalence with Jensen's is proved. This family involves additional variables, increasing then its computational complexity. Nevertheless, it contains bounds depending affinely on the measure of the interval of integration and remaining well-posed when the measure of interval of integration tends to 0. This is of great interest when numerical tools are sought, e.g. LMIs. This tightness and structural advantages prove their efficiency and motivate their use, e.g. for the analysis of time-delay and sampled-data systems.

The paper is structured as follows, Section \ref{sec:1} is devoted to the conservatism analysis of the Jensen's inequality through the use of the Gr\"{u}ss inequality. In Section \ref{sec:2}, the fragmentation procedure is studied. Finally, Section \ref{sec:3} concerns the derivation of a family of bounds equivalent to Jensen's and the comparison to existing bounds of the literature.

Most of the notations are standard except maybe for $\col_{i=1}^N(u_i)$ defining the column vector $\begin{bmatrix}
  u_1^T & \ldots & u_N^T
\end{bmatrix}^T$. The identity matrix is denoted by $I$. For $f,g\in\mathbb{R}^n$, we denote componentwise inequalities by $f\le g$. For Hermitian matrices $P$ and $Q$, $P\prec Q$ (resp. $P\preceq Q$) stands for $P-Q$ negative definite (resp. negative semidefinite).

In stability theory for linear systems, the associated convex function in (\ref{eq:Jensenc}) and (\ref{eq:Jensend}) is very often $\phi_Q(z)=z^TQz$ with $Q=Q^T\succ0$. But letting $\tilde{z}=Lz$ with $Q=L^TL$, we get $\phi_Q(z)=\phi_I(\tilde{z})$. Hence, throughout the paper we will consider the function $\phi\equiv\phi_I$ without loss of generality.

\section{Conservatism of the Jensen's inequality}\label{sec:1}

\subsection{The Gr\"{u}ss inequality \cite{Gruss:35}}

Let us consider a general inner product space $\mathcal{H}:=(\mathcal{L}(U,\mathbb{R}),\langle\cdot,\cdot\rangle_U)$ over $\mathbb{R}$. A simple, but sufficient for our problem, version of the Gr\"{u}ss Inequality on inner product spaces \cite{Dragomir:99} is defined as follows:
\begin{lemma}[Gr\"{u}ss Inequality]\label{def:gruss}
  Assume there exist bounded $f^-,f^+,g^-,g^+\in\mathbb{R}$ such that $f^-\le f^+$, $g^-\le g^+$ and functions $f,g\in\mathcal{L}(U,\mathbb{R})$ satisfying $f^-\le f\le f^+$ and $g^-\le g\le g^+$ almost everywhere on $U$. Then the following inequality
  \begin{equation*}
    \left|\mu(U)^{-1}\langle f,g\rangle_U-\mu(U)^{-2}\langle f,\mathbbm{1}\rangle_U\langle g,\mathbbm{1}\rangle_U\right|\le\frac{1}{4}\delta_f\delta_g
  \end{equation*}
  holds with $\delta_f=f^+-f^-$, $\delta_g=g^+-g^-$ and $\mathbbm{1}(\cdot)=1$ on $U$, 0 otherwise. Moreover the constant term $1/4$ in the right-hand side is sharp and is obtained for the functions $f(s)=g(s)=\sgn\left(s-(a+b)/2\right)$ where $\sgn(\cdot)$ is the signum function, $a=\inf\{s\in U\}$ and $b=\sup\{s\in U\}$.
\end{lemma}

More general versions of the Gr\"{u}ss inequality can be found in \cite{Dragomir:04} and references therein, especially for complex functions, more general measure spaces or inner product spaces.

\subsection{Conservatism of the Jensen's inequality}
In the continuous case, the function space is defined as:
\begin{equation*}
  \mathcal{L}_c(U,\mathbb{R}):=\left\{f:U\to\mathbb{R},\ f\ \mathrm{bounded, measurable\ on\ }U\right\}
\end{equation*}
where $U$ is a connected bounded set of $\mathbb{R}$. Let us also define the Hilbert space $\mathcal{H}_c:=(\mathcal{L}_c,\langle\cdot,\cdot\rangle)$ with inner product
$$\langle f,g\rangle:=\int_Uf(s)g(s)d\mu(s)$$
where $\mu$ is a nonnegative measure. Using the Gr\"{u}ss inequality (Lemma \ref{def:gruss}), the following result on the Jensen's inequality gap is obtained:
\begin{theorem}\label{th:th1}
   Given a function $f\in\mathcal{L}_c(U,\mathbb{R}^n)$ and the convex function $\phi(z)=z^Tz$, then the Jensen's gap verifies
   \begin{equation}\label{eq:th1}
   \begin{array}{l}
          \mu(U)\int_U\phi(f(s))d\mu(s)-\phi\left(\int_Uf(s)d\mu(s)\right)\le\frac{\mu(U)^2}{4}\phi(\delta_f)
   \end{array}
   \end{equation}
   where $f^-\le f(\cdot)\le f^+$ almost everywhere on $U$. Moreover the constant term $1/4$ in the right-hand side in sharp and is obtained for the functions $f_i(s)=\sgn\left(s-(a+b)/2\right)$, $i=1,\ldots,n$ with $a=\inf\{s\in U\}$ and $b=\sup\{s\in U\}$.
\end{theorem}
\begin{proof}
Let us consider the Jensen's inequality (\ref{eq:Jensenc}) with $\phi(z)=z^Tz$ and $f\in\mathcal{L}_c(U,\mathbb{R}^n)$. Simple calculations yield:
\begin{equation*}
\begin{array}{lcl}
    \mathcal{I}_1&:=&\int_Uf(s)^Tf(s)d\mu(s)\\
    &=&\sum_i\langle f_i,f_i\rangle\\
    \mathcal{I}_2&:=&\int_Uf(s)^Td\mu(s)\int_Uf(s)d\mu(s)\\
    &=&\sum_i\langle f_i,\mathbbm{1}\rangle\langle f_i,\mathbbm{1}\rangle.
\end{array}
\end{equation*}
Note that the Jensen's inequality (\ref{eq:Jensenc}) writes $\mathcal{I}_2\le \mu(U)\mathcal{I}_1$, hence we examine the difference $\mu(U)\mathcal{I}_1-\mathcal{I}_2$ and the application of the Gr\"{u}ss inequality yields
$0\le\mu(U)\mathcal{I}_1-\mathcal{I}_2\le \frac{\mu(U)^2}{4}\phi(f^+-f^-).$ The proof is complete.
\end{proof}

We have the following corollary when the continuous function $f$ is differentiable almost everywhere:
\begin{corollary}\label{cor:cor1}
   Given a continuous function $f\in\mathcal{L}_c(U,\mathbb{R}^n)$ differentiable almost everywhere and the convex function $\phi(z)=z^Tz$; then the Jensen's gap (\ref{eq:th1}) is bounded from above by
   $\frac{\mu(U)^4}{4}\phi\left(\sup_{U}|f^\prime|\right)$  where\footnote{The derivative has to be understood here in a weak sense due to the possible presence of nonsmooth points. Hence the supremum at these points is taken over all possible values for the derivative.} $\sup_{U}|f^\prime|=\col_i[\sup_{U}\{|f_i^\prime|\}]$. However, in such a case the coefficient 1/4 is not sharp anymore since the differentiability condition is not taken into account in the derivation of the Gr\"{u}ss inequality.
\end{corollary}
\begin{proof}
  If the function is differentiable almost everywhere then we have $f_i^+-f_i^-\le \mu(U)\sup_U\left|f^\prime_i\right|$. The substitution inside (\ref{eq:th1}) yields the result.
\end{proof}

\begin{remark}
  A discrete-time counterpart of Theorem \ref{th:th1} is easy to obtain. Both the upper bound and the worst-case function $f$ share a very similar expression with the continuous-time ones. It is however important to note that many works have been devoted to the obtention of tight upper bounds for the discrete Jensen's inequality, see e.g. \cite{Dragomir:04,Simic:09} and references therein.
\end{remark}

\section{Jensen's Bound Gap Reduction}\label{sec:2}

From the results of Theorem \ref{th:th1} and Corollary \ref{cor:cor1}, it is easily seen that the gap depends on the measure of the set $U$ and on the variability of the function $f$. Since the constant term is sharp, at least for discontinuous functions, this means that the inequality is conservative when the integration domain is large. In \cite{Gouaisbaut:06, Han:08, Briat:08phd}, where time-delay systems are studied, a fragmentation of the integration domain is performed and this procedure refined the delay margin estimation. In the following, we will theoretically show, in the continuous-time, that the fragmentation procedure does asymptotically reduce the gap of the Jensen's inequality to 0. It will be also proved that the gap upper bound converges sublinearly. Finally, the convergence speed can be increased using non-uniform fragmentation schemes. Although, we will only consider the continuous-time case, the same conclusions can be drawn for the discrete-time case.

\subsection{General results on gap reduction by fragmentation}

To partially explain this in the general continuous case, let us introduce the integrals:
\begin{equation*}
\begin{array}{lcl}
  \mathcal{S}(U)&:=&-\int_U[\phi\circ f](s)d\mu(s)\\
  \mathcal{S}_i(U)&:=&-\int_{U_i}[\phi\circ f](s)d\mu(s),\ i=1,\ldots,N
\end{array}
\end{equation*}
where the connected sets $U_i$'s satisfy $\bigcup_{k=1}^NU_k=U$ and $\mu(U_j\cap U_i)=0$ for all $i,j=1,\ldots,N$, $i\ne j$. In such a case, we have $\mathcal{S}(U)=\sum_i\mathcal{S}_i(U)$. Then, rather than bounding $\mathcal{S}(U)$, each $\mathcal{S}_i(U)$ will be bounded separately by $\mathcal{J}_i$ and the respective bounds added up. This yields the following result:

\begin{theorem}\label{th:th1r}
  Let us assume that the compact and connected set $U\subset\mathbb{R}$ is partitioned in $N$ disjoint parts as described above. Let us also consider the functions $\phi(z)=z^Tz$ and $f\in\mathcal{L}_c$. In such a case, the Jensen's gap is bounded from above by:
  \begin{equation*}
    \mathcal{S}(U)-\sum_{i=1}^N\mathcal{J}_i\le e_1(N)=\frac{1}{4}\sum_{j=1}^N\mu(U_j)^2\phi(M^j-m^j)
  \end{equation*}
  where $m^j=\col_i[m_i^j]$, $m_i^j:=\inf_{U_j}\{f_i\}$, $M^j=\col_i[M_i^j]$, $M_i^j:=\sup_{U_j}\{f_i\}$. Note that the term $1/4$ is also sharp in this case.
\end{theorem}

\begin{proof}
  The proof is similar as for Theorem \ref{th:th1}.
\end{proof}

\begin{corollary}\label{cor:cor2r}
   Let us assume that the compact and connected set $U\subset\mathbb{R}$ is partitioned in $N$ disjoint parts as described above. Given a continuous function $f\in\mathcal{L}_c(U,\mathbb{R}^n)$ which is differentiable almost everywhere and the convex function $\phi(z)=z^Tz$; then the Jensen's gap is bounded from above by:
   \begin{equation*}
   \mathcal{S}(U)-\sum_{i=1}^N\mathcal{J}_i\le\ e_2(N)=\frac{1}{4}\sum_{j=1}^N\mu(U_j)^4\phi\left(\sup_{U_j}|f^\prime|\right)
   \end{equation*}
   where $\sup_{U_j}\{|f^\prime|\}=\col_i[\sup_{U_j}\{|f^\prime_i|\}]$ (in a weak sense).
\end{corollary}

\begin{remark}
  More general versions can also be defined using $\phi_Q$ instead of $\phi_I$. Fragment-dependent $\phi_{Q_j}$ can also be considered; see e.g. \cite{Gouaisbaut:06, Han:08, Briat:08phd}. The analysis follows the same lines and the results are qualitatively identical, this is thus omitted.
\end{remark}

\subsection{Equidistant fragmentation}\label{sec:eqdsdqsd}

Let us consider the most simple case where the Lebesgue measure is considered together with a fragmentation of $U$ in $N$ parts of identical measure. Then the following corollary can be derived.
\begin{corollary}\label{cor:cor2s}
  Assume that $f\in\mathcal{L}_c$ satisfies the assumptions of the previous results. Fragmenting the set $U$ in $N$ parts of identical Lebesgue measure yields the following bound for the Jensen's gap:
  \begin{equation}\label{eq:e1}
    e_1(N)=\frac{\mu(U)^2}{4N}\phi(\theta)
\end{equation}
where $\theta=\col_i(\theta_i)$ and $\theta_i=\max_j\{M_i^j-m_i^j\}$. Moreover, when a continuous function $f$ differentiable almost everywhere is considered, we get the bound:
\begin{equation}\label{eq:e2}
    e_2(N)=\frac{\mu(U)^4}{4N^3}\phi(\eta).
\end{equation}
where $\eta=\max_j\left\{\sup_{U_j}|f^\prime|\right\}$.
\end{corollary}

The following proposition provides the result on the sublinear convergence of the gap upper bound when the fragmentation order increases:
\begin{proposition}
  When no continuity assumption is made on the function $f\in\mathcal{L}_c$, the upper bound $e_1(N)$ on the gap satisfies
  \begin{equation*}
    \begin{array}{lcl}
      e_1(N)&=&\left(1-N^{-1}\right)e_1(N-1),\ N>1\\
      e_1(1)&=&\frac{\mu(U)^2}{4}\phi(\theta).
    \end{array}
  \end{equation*}

  When the continuous function $f\in\mathcal{L}_c$ is differentiable almost everywhere, the upper bound $e_2(N)$ on the gap obeys
  \begin{equation*}
    \begin{array}{lcl}
      e_2(N)&=&\left(1+\frac{-3N^2+3N-1}{N^3}\right)e_2(N-1),\ N>1\\
      &\sim&\left(1-3N^{-1}\right)e_2(N-1)\ \mathrm{when\ }N\to+\infty\\
      e_2(1)&=&\frac{\mu(U)^4}{4}\phi(\eta).
    \end{array}
  \end{equation*}
  \end{proposition}

  Hence, in both cases, the error tends asymptotically to 0 and the convergence is sublinear since $\lim_{N\to+\infty}\frac{e_i(N+1)}{e_i(N)}=1$, $i\in\{1,2\}$. Hence we can easily conclude that
  \begin{equation*}
    \sum_{i=1}^N\mathcal{J}_i\to\mathcal{S}(U)\ \mathrm{when}\ N\to+\infty
  \end{equation*}
  in any case.
%

\begin{example}\label{ex:2}
  Let us consider the integral
  \begin{equation*}
    J(\alpha):=-\int_0^1e^{2\alpha t}dt=\frac{1}{2\alpha}(1-e^{2\alpha})
  \end{equation*}
  for some $\alpha\in\mathbb{R}$ (by continuity we have $J(0)=-1$). Now consider the following sum of Jensen's bounds taken over each $[ih,(i+1)h]$ with $h=1/N$:
  \begin{equation*}
  \begin{array}{lcl}
        J_N(\alpha)&:=&-\frac{1}{h}\sum_{i=0}^{N-1}\left(\int_{ih}^{(i+1)h}e^{\alpha t}dt\right)^2\\
        &=&-\frac{N(1-e^{\alpha/N})(1-e^{2\alpha})}{\alpha^2(1+e^{\alpha/N})}\\
        &\to& \frac{1}{2\alpha}(1-e^{2\alpha})\ \mathrm{as\ }N\to+\infty.
  \end{array}
  \end{equation*}
  This shows the asymptotic convergence. To see the sublinear convergence, it is enough to remark that
  $$\lim_{N\to+\infty}\frac{J_{N+1}(\alpha)-J(\alpha)}{J_{N}(\alpha)-J(\alpha)}=1.$$ 

  In Fig. \ref{fig:rate}, we can see the evolution of the normalized bound $J_N(\alpha)/J(\alpha)$ for different values for $\alpha$. For small positive value for $\alpha<1$, the convergence is very fast since the function is slowly varying. When $\alpha$ increases the convergence becomes slower. This follows from Theorem \ref{th:th1} and Corollary \ref{cor:cor1} stating that the gap is depends quadratically on the variability of the function. In Fig. \ref{fig:lolmdr}, the different upper bounds $e_1(N)$ (\ref{eq:e1}) and $e_2(N)$ (\ref{eq:e2}) are compared to the actual gap for the case $\alpha=1$.

  \begin{figure}
  \begin{center}
    \includegraphics[width=0.35\textwidth]{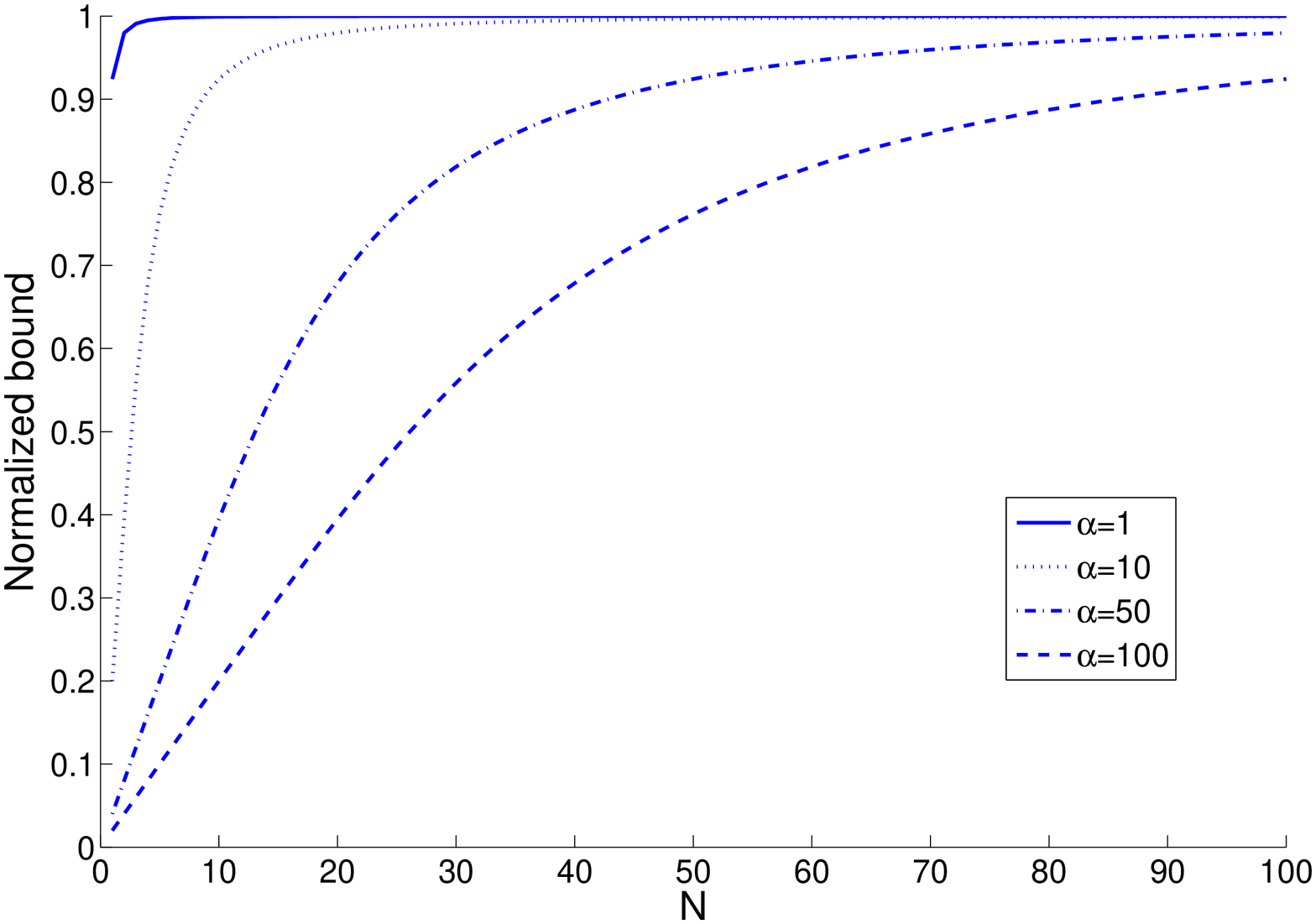}\hfill\includegraphics[width=0.35\textwidth]{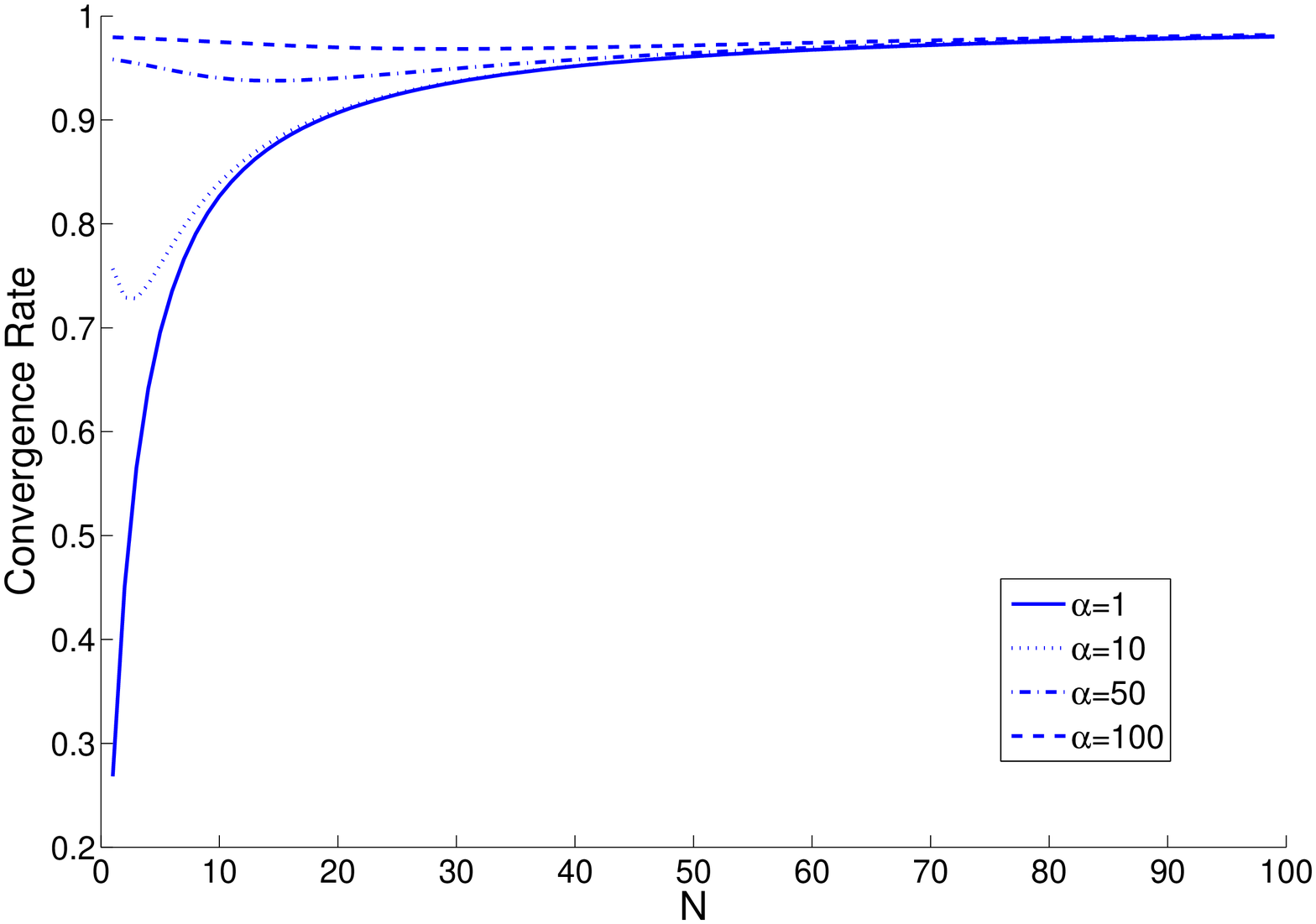}
    \caption{Evolution of the normalized bound $J_N(\alpha)/J(\alpha)$ (top) and the ratio $J_{N+1}(\alpha)/J_N(\alpha)$ (bottom) for different values for $\alpha$}\label{fig:rate}
    \end{center}
  \end{figure}
    \begin{figure}
  \begin{center}
    \includegraphics[width=0.35\textwidth]{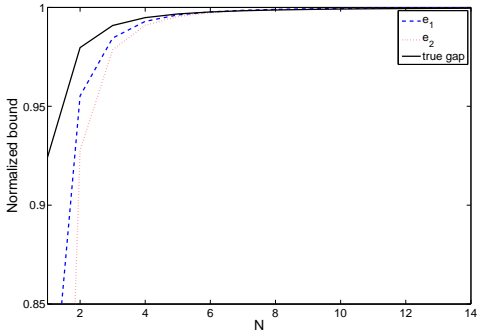}
    \caption{Evolution of the normalized bound $J_N(1)/J(1)$ and comparison with upper bounds $e_1(N)$ (\ref{eq:e1}) and $e_2(N)$ (\ref{eq:e2})}\label{fig:lolmdr}
    \end{center}
  \end{figure}
  \meth
\end{example}

\subsection{Nonuniform fragmentation}

The second conclusion, difficult to consider when analyzing time-delay systems, concerns the fact that an adaptive fragmentation scheme could improve the efficiency of the method. Indeed, defining fragments whose measure is inversely proportional to the variability of the function should reduce the gap more efficiently than the naive equidistant fragmentation. This is an immediate consequence of the fact that Jensen's inequality is an equality for the set of constant functions $f$ (i.e. $f^-=f^+$). 

\begin{example}
  Let us illustrate the above discussion by considering the critical function:
  $$f(s)=\sgn\left(s-(a+b)/2\right),\quad s\in[a,b]$$
  and the Lebesgue measure $\mu$. Define also the intervals $U_1(\epsilon)=\left[a,(a+b-\epsilon)/2\right]$, $U_2(\epsilon)=\left[(a+b-\epsilon)/2,(a+b+\epsilon)/2\right]$ and $U_3(\epsilon)=\left[(a+b+\epsilon)/2, b\right]$ for some $\epsilon\in(0,a+b)$. It is clear that $\bigcup_{k=1}^3U_k(\epsilon)=[a,b]$ and $\mu(U_i(\epsilon)\cap U_j(\epsilon))=0$, for all $i,j\in\{1,2,3\}$, $i\ne j$ and for any $\epsilon\in(0,a+b)$. Since the function $f(s)$ is constant over $U_1(\epsilon)$ and $U_2(\epsilon)$, for any $\epsilon\in(0,a+b)$, then the Jensen's inequality is exact and does not introduce any conservatism. All the conservatism is concentrated on the interval $U_2(\epsilon)$ where lies the discontinuity. Finally, using Theorem \ref{th:th1}, the exact gap on this interval is
  \begin{equation*}
    \frac{\mu(U_2(\epsilon))^2}{4}\phi(f^+-f^-)=\epsilon^2.
  \end{equation*}
  Thus the gap can be reduced to an arbitrary small value by choosing adequately the set $U_2(\epsilon)$.\meth
\end{example}

It is important to note that, when a uniform fragmentation scheme is used on the above discontinuous function, the gap does not converge monotonically. Indeed, by increasing $N$, the measure of the interval where lies the discontinuity can be locally increasing. The non-monotonic gap is however bounded from above by the monotonic bounds $e_1$ and $e_2$ derived in Section \ref{sec:eqdsdqsd}.

\begin{example}
  Let us consider the function of Example \ref{ex:2}. The idea is to use a nonuniform fragmentation to speed up the convergence. Since the slope of the function increases, then it seems natural to fragment the interval $[0,1]$ in such a way that the measure of the fragments decreases as we approach 1. We thus consider the following delimitating sequence of points $t_i:=(1-\varepsilon^{\frac{i-1}{N}})(1-\varepsilon)^{-1}$ where $\varepsilon\in(0,1)$ is a small positive scalar and $i=1,\ldots,N+1$. Obviously we have $t_1=0$, $t_{N+1}=1$ and the interval $[0,1]$ is nonuniformly partitioned in $N$ parts $U_i:=[t_i,t_{i+1}]$, $i=1,\ldots,N$. The Lebesgue measure of the interval $U_{i}$ satisfies $\mu(U_i)=\varepsilon^{\frac{i}{N}}\kappa_0$
  where $\kappa_0=\frac{\varepsilon^{-\frac{1}{N}}-1}{1-\varepsilon}$. Choosing $\varepsilon=10^{-4}$ and considering the exponential function of Example \ref{ex:2} with $\alpha=100$, we obtain the result depicted on Fig. \ref{fig:2} where we can see that the convergence speed has been increased quite spectacularly.
  \begin{figure}
    \begin{center}
     \includegraphics[width=0.35\textwidth]{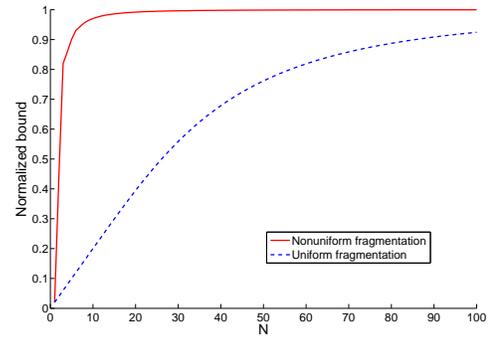}
     \caption{Comparison of normalized bounds with uniform and nonuniform fragmentation for the exponential function with $\alpha=100$ and $\varepsilon=10^{-4}$}\label{fig:2}.
    \end{center}
  \end{figure}\meth
\end{example}

Unfortunately, despite of being very efficient for the exponential function, this is not of real interest for the analysis of time-delay and sampled-data systems since the trajectories of the system are not known a priori. To explain this, let us consider the time delay system $\dot{x}(t)=-3/2\cdot x(t-1)$ with functional initial condition $x(\theta)=1,\ \theta\in[-1,0]$ whose unique solution is oscillating and exponentially stable. Choose two different time instants $t_1<t_2$ and introduce the intervals $I_1:=[t_1-1,t_1]$, $I_2:=[t_2-1, t_2]$. Assuming that the exact solution of the system is known, then an adapted nonuniform partitioning of $I_1$ can be constructed. However, this partitioning fails almost surely to be a good one for the interval $I_2$ due to the oscillating behavior of the solution. This shows that even when the solution is known, it is, in general, not possible to find a good nonuniform fragmentation common to any interval of integration. Hence, it is natural to choose a \emph{uniform fragmentation which is the best tradeoff between all the non-uniform fragmentation schemes}.
%
%
%

\section{Equivalence between Jensen's bound and some bounds of the literature}\label{sec:3}

The goal of this section is to derive a complete family of bounds equivalent to Jensen's in terms of tightness but more complex from a computational point of view. Despite of their slight higher computational cost, this family has the nice properties of being affine in the measure of the interval of integration and leading to LMIs that remain well-posed when the measure of the interval of integration tends to 0. This is very convenient when this quantity is a (time-varying or uncertain \cite{Naghshtabrizi:08,Seuret:09b}) data of the problem\footnote{This is, for instance, the case when time-delay systems are studied. In such a case, the length of the interval of integration coincides with the delay itself.}. The latter feature is due to the convexity (affine) of the affine bound w.r.t. the measure of the interval integration. This is very interesting when LMI-based results are desired as it will be illustrated in Section \ref{sec:motivationlol}. The upcoming results can be used to motivate the use of such affine bounds which are not worse than Jensen's in terms of tightness. It is finally shown that several bounds devised in the literature are elements of this general family.

The results of this section rely on the following lemma:
\begin{lemma}[\cite{Briat:08phd}]\label{lem:equiv}
Given matrices $C=C^T\succ0$, $A=A^T$ and $B$, the following statements are equivalent:
  \begin{enumerate}
    \item The matrix inequality
    \begin{equation}\label{eq:loil}
      \mathcal{M}_1:=A-B^TC^{-1}B\prec0
    \end{equation}
    holds.
    \item There exists a matrix $N$ such that the matrix inequality
    \begin{equation}\label{eq:nl}
      \mathcal{M}_2(N):=A+N^TB+B^TN+N^TCN\prec0
    \end{equation}
    holds.
    \item The statement
    \begin{equation*}
      \inf_N\left\{N^TB+B^TN+N^TCN\right\}=-B^TC^{-1}B.
    \end{equation*}
     holds true in the partially ordered space of symmetric matrices with partial order '$\prec$'. Moreover the global minimizer $N^*$ is unique and is given by $N^*=-C^{-1}B$.
  \end{enumerate}
\end{lemma}

\begin{proof}
  A proof is given in \cite{Briat:08phd} and is quite involved.  Here we provide an alternative one (some other proofs could rely on the elimination/projection lemma). To see the equivalence between 1) and 2) is enough to show that 3) holds. It is easy to see that (\ref{eq:nl}) is convex in $N$ since $R\succ0$. Hence completing the squares, we find that the minimum $-B^TC^{-1}B$ is attained for $N^*=-C^{-1}B$. Thus, for any triplet $(A,B,C)$ satisfying the assumptions, we can always find $N^*$ such that $$\mathcal{M}_1=\mathcal{M}_2(N^*)=\inf_N\left\{\mathcal{M}_2(N)\right\}.$$ This concludes the proof.
\end{proof}

%
%

The interest of the above result is twofold: it can be used to transform complex nonlinear matrix inequalities \cite{Moon:01, Briat:08c, Peaucelle:05} in a more convenient form \cite{Briat:08c}; or, what is of interest here, to prove equivalence between different results.
%
%
This is stated in the following theorem:
\begin{theorem}\label{th:mainthc}
  Let us consider a vector function $z(\cdot)$ integrable over $U$, with Lebesgue measure $\mu(U)$, a real matrix $R=R^T\succ0$, and a vector function $w(\cdot)$ verifying  $\int_Uz(s)ds=Mw(\cdot)$, for some known matrix $M$. Then the following statements are equivalent:
  \begin{enumerate}
    \item The following inequality
     \begin{equation*}
    -\int_Uz(s)^TRz(s)ds\le-\mu(U)^{-1}w^TM^TRMw
  \end{equation*}
  holds for all $w(\cdot),z(\cdot)$ satisfying the above assumptions.
  \item There exists a matrix $N$ such that the inequality
  \begin{equation*}
  -\int_Uz(s)^TRz(s)ds\le w^T\mathcal{Q}(N)w
\end{equation*}
holds for all $w(\cdot),z(\cdot)$ satisfying the above assumptions and where $$\mathcal{Q}(N)=N^TM+M^TN+\mu(U)N^TR^{-1}N.$$
\end{enumerate}
\end{theorem}
\begin{proof}
  The proof is a consequence of Lemma \ref{lem:equiv}.
\end{proof}
\begin{remark}
  A discrete-time formulation can be obtained in the same way. This is omitted due to space limitations.
\end{remark}
%

In the sequel, we will apply Theorem \ref{th:mainthc} and its discrete-counterpart in order to show the equivalence between different results of the literature.

\subsection{A first Integral inequality}

Let us consider a differentiable function $x(t)$ verifying $\int_{t_k}^t\dot{x}(s)ds=Mw$ with $M=\begin{bmatrix}
  I & -I
\end{bmatrix}$ and $w(\cdot)=\col(x(t),x(t_k))$. In \cite{Seuret:09b}, the following bound is used:
\begin{equation}\label{eq:seurette}
  -\int_{t_k}^tz(s)^TRz(s)ds\le w^T\mathcal{R}w,\qquad R=R^T\succ0,\ t>t_k
\end{equation}
where $\mathcal{R}=N^TM+M^TN+(t-t_k)N^TR^{-1}N$ and $N$ is an additional matrix to be determined. Then according to Theorem \ref{th:mainthc}, we can conclude on the equivalence with Jensen's. However, we will see in the next example that it is sometimes better suited to use the affine formulation.

\subsection{A reason for using the affine formulation rather the rational one}\label{sec:motivationlol}

This discussion aims at illustrating the ill-posedness problem arising when the support of the integral varies in time and may vanish at some instants. The affine formulation does remain well-posed in such circumstances leading then to more appropriate numerical tools, like LMIs. In \cite{Seuret:09b}, aperiodic sampled-data systems are considered and an affine version of the Jensen's inequality is employed to provide an LMI condition \cite[Theorem 1]{Seuret:09b}. If the rational one was used, this would create a concave term in $(t-t_k)$ of the form $-(t-t_k)^{-1}M_2^TRM_2$ where $R=R^T\succ0$, $M_2\in\mathbb{R}^{n\times 2n}$, $t\in[t_k,t_{k+1}]$, $t_{k+1}-t_k\le\tau_m$, $t_k$ being the sampling instants (following the notation of \cite{Seuret:09}). This term is ill-posed when $t\to t_k$ and a way to overcome this problem consists of bounding this term by $-\tau_m^{-1}M_2^TRM_2$. We compare now this 'result' to the Theorem 1 of \cite{Seuret:09b} on the system \cite[Example 4]{Naghshtabrizi:08}, \cite[Example 1]{Seuret:09b}.
Theorem 1 of \cite{Seuret:09b}, based on the affine formulation, yields a maximal $\tau_m=1.6894$ while the 'result' based on the rational Jensen's inequality yields the lower value $\tau_m=0.8691$. Even though the bounds are initially equivalent, \emph{the desire of making the problem tractable (obtaining well-posed LMIs) introduces considerable conservatism}. This illustrates the importance of the affine version of the Jensen's inequality since we have to favor tools in calculations that lead to better numerical solutions.

\subsection{A second Integral inequality}

In \cite[equation (7)]{ZhangHan:08}, the following bound is considered:
\begin{equation}\label{eq:han}
  -\int_{t-\tau}^t\dot{x}(s)^TR\dot{x}(s)ds\le\xi(t)^T(\mathscr{M}+\tau N^TR^{-1}N)\xi(t).
\end{equation}
Simple calculations on this upper bound shows that Theorem \ref{th:mainthc} applies and that this bound is equivalent to the Jensen's inequality.
%

\subsection{A sum inequality}

In \cite[equation (7)]{Zhang:06}, a bound of the form is introduced:
\begin{equation}\label{eq:salechinois}
  -\sum_{i=k-h}^{k-1}y(i)^TRy(i)\le\xi(k)^T\left[\beth+hN^TR^{-1}N\right]\xi(k).
\end{equation}
In this case, the discrete-time version of Theorem \ref{th:mainthc} can be applied, showing then equivalence with the discrete-time version of the Jensen's inequality.

\section{Conclusion}

The conservatism of the Jensen's inequality has been analyzed using the Gr\"{u}ss inequality. Motivated by several results of the literature, a fragmentation scheme has been considered. It has been shown that the gap converges asymptotically to 0 as the order of fragmentation increases. Next, nonuniform fragmentation techniques have been introduced and their possible accelerating effect illustrated. Unfortunately, they can be applied in some very specific cases only. This showed that the best tradeoff lies in the use of uniform fragmentation schemes.

The second part of the paper has been devoted to the characterization of a family of bounds, equivalent to Jensen's in terms of tightness but with a higher computational complexity. This family defines affine bounds in the measure of interval of integration (rational and nonconvex for Jensen's) for which the obtained matrix inequalities remain well-posed when the measure of the integral of integration tends to 0. This is of crucial interest when LMIs are sought. It has been shown that several bounds devised in the literature are elements of this family.

As a final remark, this homogeneity suggests that Jensen's inequality and its companions could be the best bounds still preserving a tractable structure to the problem. This together with a (possibly adaptive) fragmentation scheme should lead to asymptotically exact well-posed approximants of integral terms, affine in the measure of the integration support.

\section*{Acknowledgments}

The author thanks the ACCESS team, the editor and the associate editor as well as the anonymous reviewers who helped to improve the quality of the paper.

\bibliographystyle{IEEEtran}
\bibliography{./global}

\end{document}